%% file: SocialNetworkCascades.tex
\documentclass[11pt]{article} 

\usepackage{amsmath,amssymb,amsthm,amsfonts,latexsym,bbm,xspace,graphicx,float,mathtools}
\usepackage{color}
\usepackage{algorithm}
\usepackage{algorithmic}

\newcommand{\smw}[1]{{\color{black}{#1}}}

\usepackage{geometry} 
\usepackage{graphicx} 

\usepackage{array} 
\usepackage{paralist} 
\usepackage{verbatim} 
\usepackage{subfig} 

\newtheorem{theorem}{Theorem}

\newtheorem{corollary}{Corollary}
\newtheorem{lemma}{Lemma}

\newtheorem{claim}{Claim}
\newtheorem{proposition}{Proposition}

\newtheorem{definition}{Definition}

\newtheorem{observation}{Observation}

\begin{document}

\title{Reaching Consensus via non-Bayesian Asynchronous Learning in Social Networks}
\author{Michal Feldman\footnote{Tel-Aviv University, michal.feldman@cs.tau.ac.il. Michal Feldman is partially supported by the European Research Council under the European Union's Seventh Framework Programme (FP7/2007-2013) / ERC grant agreement number 337122.} 
\and
Nicole Immorlica\footnote{Microsoft Research, nicimm@microsoft.com} 
\and
Brendan Lucier\footnote{Microsoft Research, brlucier@microsoft.com} 
\and
S. Matthew Weinberg\footnote{MIT, smweinberg@csail.mit.edu. Supported by a Microsoft Research Fellowship.}
}
\date{}

\newcommand{\E}{\mathbb{E}}
\maketitle
\begin{abstract}
\input{abstract.tex}
\end{abstract}
\thispagestyle{empty}

\input{SNC-intro.tex}
\input{SNC-prelim.tex}
\input{SNC-examples.tex}
\input{SNC-main.tex}

\input{SNC-conclusion.tex}

\bibliographystyle{plain}
\bibliography{SNC-bib}

\appendix
\input{SNC-hardexamples.tex}

\end{document}

%% file: abstract.tex
We study the outcomes of information aggregation in online social networks.  Our main result is that networks with certain realistic structural properties avoid information cascades and enable a population to effectively aggregate information.  In our model, each individual in a network holds a private, independent opinion about a product or idea, biased toward a ground truth.  Individuals declare their opinions asynchronously, can observe the stated opinions of their neighbors, and are free to update their declarations over time.  Supposing that individuals conform with the majority report of their neighbors, we ask whether the population will eventually arrive at consensus on the ground truth.  We show that the answer depends on the network structure: there exist networks for which consensus is unlikely, or for which declarations converge on the incorrect opinion with positive probability.  On the other hand, we prove that for networks that are sparse and expansive, the population will converge to the correct opinion with high probability.

%% file: SNC-intro.tex

\section{Introduction}

A community consists of a collection of individuals, each with their own observations and inferences.  Through social interactions, these individuals combine these private reflections with the public opinions of others to form their personal public opinions regarding matters of importance.  For many such matters, individuals have aligned goals. Thus, there is often a ground truth, a correct answer, to such questions.  When ground truth exists and when individuals' observations are more likely to lead to correct inferences than incorrect ones, the law of large numbers states that a majority of individuals, when reasoning privately, will reach correct conclusions.  This leaves a potentially substantial fraction of society with the incorrect conclusion, but it offers hope that the correct majority might influence the society creating a consensus on the ground truth.

Unfortunately, the outcome of this process of social deliberation can result in egregious errors in which the potentially small incorrect minority opinion infiltrates the entire community as individuals copy this opinion.  A situation like this, in which individuals copy opinions of others while ignoring their own observations, is called an {\it information cascade}.  Information cascades notoriously block information aggregation.  That is, although society has enough information for {\it everyone} to make the right decision with high probability, there is a substantial chance that everyone makes the wrong decision!


%
In this work we are motivated by occurrences like the following two real historical events.
In the 1930s, the United States experienced a severe drought, spurring a great innovation in agriculture: hybrid corn.
These new hybrids offered a yield $15-20\%$ greater than the open-pollinated varieties, and by the early 40s they dominated the corn belt.   Interviews with farmers regarding their adoption practices suggest that the two main factors in the acceptance or rejection of hybrid corn were personal experimentation and the opinions of friends. As farmers repeatedly weighed these factors from year to year, the farming society as a whole gradually began to herd on the highly beneficial decision to plant hybrid corn.
%

In the late 2000s, the United States experienced another period of economic decline that has come to be known as the Great Recession.  The cause of the recession is commonly attributed to the collapse of the housing bubble.  Economists have argued that, again, a main factor in investors' actions in the context of the housing market was the investment decisions of others.
Thus, again the individuals in the community herded on a certain behavior albeit this time a suboptimal one.

How is it that both communities -- farmers in the 30s and 40s, and investors in the 2000s -- reached agreement on the answer to important questions facing them?  How is it that the farmers reached the correct conclusion while the bankers were fooled \emph{en masse}?  A crucial difference between these two cases is the structure of the network over which information spreads.  Farmers live in local communities and mainly interact with geographically close neighbors whereas investors observe investment decisions of most others.

In the present day, the proliferation of online social platforms such as Facebook and Twitter serves to remove friction in the dissemination of information.  One might expect that adoption of new technologies or opinions, in the spirit of hybrid corn or housing investments, would occur at a more rapid pace as a result and have widespread impact.  This leads us to our main motivating question: does the structure of large, online social networks enable the efficient aggregation of information, while resisting the proliferation of incorrect beliefs?

An important research question is to understand the factors that influence information cascades.  What networks of social interactions, and what patterns of opinion formation allow entire societies to converge on the correct decision?  There is a long literature on the topic of social learning, focused on two different barriers for information aggregation: information suppression and information loss.  Some models, like standard rational Bayesian learning models \cite{Banerjee92,Bikchandani1992,Smith00,Banerjee04,Acemoglu11}, capture the information suppression problem.  Opinions are private and are only revealed over time.  In such a model there might never be, at any time, enough public information in the society to correctly aggregate information. The typical conclusion is that this suppression effect is worst for the complete network, and can be avoided if the network is (in some sense) sparse \cite{Banerjee04,Acemoglu11}.

In other models, such as repeated synchronous majority dynamics, agents begin by announcing their opinions publicly, so a central observer would initially be able to deduce the correct decision.  However,
since the agents use heuristics that are based on their own local view of the network to update their beliefs
(e.g., switching to the majority report of one's neighbors), the community might diverge from this state, experiencing information loss.  Indeed, there are scenarios in which a very small minority opinion can ultimately dominate the ground truth~\cite{Berger01}.  However, social learning \emph{does} occur in such models if the network is sufficiently well-connected, and no single individual is too influential \cite{Mossel2013}.

These two lines of work arrive at very different conclusions about the impact of network topology. Our work considers a setting that exhibits both barriers simultaneously.  Our model thus captures the tension between two requirements: being sparse enough to prevent information suppression, while being sufficiently connected to prevent information loss.

In our model, the decision at hand is binary, e.g., whether or not to adopt a certain technology.  There is a correct decision, and each individual has a (conditionally independent) signal regarding this decision which is more likely to be correct than incorrect (i.e. is correct with probability $1/2 + \delta$ for some $\delta > 0$).  Initially, individuals are not stating opinions (as in standard models of rational learning).  Individuals are asked to state an opinion, repeatedly and asynchronously.\footnote{To our knowledge, this is the first paper to study non-Bayesian asynchronous learning.}
When stating an opinion, individuals simply copy the majority opinion among their friends, breaking ties in favor of their private signal.  Our model therefore combines a non-Bayesian update method with the asynchronicity typical of Bayesian models.  This asynchronous model is natural in settings of local communication in a population, where the sharing of information is not globally coordinated. We ask: do these asynchronous majority dynamics result in a correct consensus with high probability, for graphs that exhibit realistic properties of large social networks?

We focus on two key features of large social networks.  First, they tend to be \emph{expansive}, meaning roughly that they do not contain very sparse cuts.  While it has been observed that small social networks tend to have sparse cuts corresponding to divisions between sub-communities, this tends not to be the case for empirically-observed large social networks \cite{LeskovecLDM08,MM11}.  Intuitively, expansiveness leads to information diffusion which allows society to reach consensus.  Second, social networks tend to be {\em sparse}.  Intuitively, sparsity should limit the rate at which a single individual's opinion can spread in the network, leading to the spread of many independent opinions, and independent opinions are good for producing correct majorities.  These two features together thus have some chance of producing a correct consensus so long as the low sparsity allows enough independent decisions to be reached before the high expansiveness takes over diffusing these opinions.

As we will show, it is not always the case that low average degree is sufficient to build a population to a correct consensus opinion. In fact, it is not even the case that this property suffices to reach a correct majority opinion. In Appendix~\ref{sec:examples2}, we provide an example of a network with constant average degree, for which the population will reach a majority on the incorrect opinion with positive probability.  The key issue in this construction is the presence of a large clique; that is, while the network is sparse in a global sense, it is not ``locally sparse'' in the sense that it contains a reasonably large dense subgraph.

Motivated by this example, we turn to stronger notions of sparsity. Specifically, we study the class of expanders with maximum degree $d$. Our main result is that for any fixed $d > 1$ and a growing family of graphs with maximum degree $d$ and sufficiently high expansiveness, the dynamics described above will reach consensus on the ground truth with high probability.

\vspace{0.1in}
{\bf Theorem (Informal):} Suppose $\{G_n\}_n$ is a growing family of graphs with maximum degree $d$, each with sufficiently large expansion as a function of $d$ and $\delta$.  Then the population will converge to consensus on the ground truth with probability $1 - o(1)$.
\vspace{0.1in}

We believe that max-degree $d$ can be relaxed to a weaker property of sparsity, such as bounded arboricity, in this theorem.
For example, in Section~\ref{sec:examples} we show that under the star topology, the population reaches a consensus on the ground truth with high probability. Yet, the class of max-degree $d$ is \emph{significantly} better understood and technically cleaner to work with.
We believe that our analysis of expanders with max-degree $d$ can be leveraged for a better understanding of convergence to consensus in more general classes of graphs, including graphs with alternative sparsity conditions.

%

\paragraph*{Our Techniques}
We prove our main result by dividing an execution of the behavior dynamics into two stages, which we analyze separately.  The first stage lasts for a linear (in $n$, the number of nodes) number of rounds, until most nodes have updated their opinions at least once.  We argue that, after the first stage ends, significantly more than half of the individuals (weighted by degree) in the network hold the ground truth as their opinion, with high probability.  This argument has two steps: first, we use results from the theory of boolean functions to establish that the expected number of nodes (weighted by degree) with the correct opinion is greater than half of the nodes in the network.  Second, to show that the number of correct opinions is concentrated around its expectation, we use the fact that the network has bounded degree.  This bounded degree implies that (with high probability) no individual will be very influential after only linearly many steps; indeed, the number of other individuals whose opinions could depend on the private signal of any given node will be small.  Hence most pairs of opinions will be independent after linearly many steps, and thus the variance of the number of correct opinions is small.

The second stage begins after most individuals have declared an opinion, and lasts until the dynamics converges.  For this stage, we use properties of expander graphs to show that if one opinion has a significant majority in the population, then this bias will be magnified as the process continues, until eventually the entire population reaches consensus.  This analysis makes use of the expander mixing lemma as well as the theory of biased random walks.  Since the second stage begins in a state where a significant majority of the population (weighted by degree) is reporting the correct opinion (from our analysis of the first stage), we conclude that the population reaches a correct majority with high probability.


While the second half of our argument shares structural similarity with \cite{Mossel2013}, the first half requires a novel approach. Specifically, because \cite{Mossel2013} studied synchronous learning, they immediately see a correct majority in round one, independent of the graph structure. Due to the asynchronous nature of our learning, showing that we ever reach a correct majority at any point during the process is technically challenging, and requires some assumptions on the graph structure (i.e., sparsity).

Note that we use the two required network properties, expansiveness and sparsity, in different parts of our analysis.  The sparsity condition is used to show that opinions are largely independent in the initial rounds of the dynamics, and hence a majority will report correctly.  The expansiveness condition is used to show that once the population reaches a clear majority, it will then quickly reach consensus on that majority opinion.

Motivated by this division of the analysis, we make a stronger conjecture that the implications of the two network properties should hold separately.  That is, we conjecture that any network with constant maximum degree leads the population to stabilize in a correct majority.
In Appendix~\ref{sec:examples2} we establish that this is indeed the case under the cycle topology.
Furthermore, we conjecture that any network with sufficiently high expansion will stabilize in a consensus (not necessarily a correct consensus).  




\subsection{Related Work}

Our work is related to a line of literature concerning the aggregation of information under Bayesian learning.  In the standard learning model, individuals are fully rational and are given noisy signals correlated with a ground truth.  The individuals sequentially report a ``best guess'' at the ground truth.  It was first observed by Banerjee \cite{Banerjee92} and Bikhchandani, Hirshleifer, and Welch \cite{Bikchandani1992} that a population may fail to aggregate information when reports are publicly observed, due to information cascades.  Smith and Sorensen \cite{Smith00} show that such information cascades can be avoided under the assumption that signals can be arbitrarily informative; i.e., that the strengths of agents' beliefs are unbounded.  In a spirit closer to our work, Banerjee and Fudenberg \cite{Banerjee04} suppose that each agent observes a random subset of the previous agents' actions, and show that asymptotic learning occurs whenever no agent is too influential (i.e., no agent is observed too often).  Acemoglu, Dahleh, Lobel, and Ozdaglar \cite{Acemoglu11} show that learning occurs under significantly more general conditions if agents are aware not only of which prior agents they observe, but also the entire history of prior agent observations.

An alternative line of work on social learning concerns the performance of non-Bayesian, heuristic methods of aggregating information.  In the classic model of DeGroot \cite{DeGroot74}, each agent's signal is a real number in the unit interval.  In each round, agents update their reports by taking a weighted average of their neighbors' reports.  Such a process must necessarily converge to a consensus with each connected component of a network.  Golub and Jackson \cite{Golub10} consider the question of whether this consensus agrees with the initial ground truth.  They find that this occurs if and only if the most influential (i.e., highest-degree) node is vanishingly influential as the population grows large.  These models assume a continuous space of opinions and reports.  In the case of discrete opinions, where reports are updated by taking the majority report of one's neighbors, Berger \cite{Berger01} shows that it is possible for an initial state with a constant-sized minority to lead ultimately to global adoption of the minority opinion.

The work most similar in spirit to the present paper is Mossel, Neeman, and Tamuz \cite{Mossel2013}.  They consider repeated simultaneous majority dynamics starting from an initial state in which each node takes opinion $0$ or $1$ independently at random, biased toward $1$ (the ground truth).  They study conditions under which a majority of the population reports $1$ once the dynamics converges; they show that this occurs if the graph is ``almost'' vertex transitive (in the sense that each vertex can be mapped to many other nodes by graph automorphisms).  They also show that if the graph is an expander, then majority dynamics will result in consensus with high probability.  Tamuz and Tessler \cite{Tamuz2013} derive sufficient conditions under which the ground truth can be reconstructed from the final state of the dynamics by any means, not necessarily by taking the majority report of the population.

The crucial difference between this line of work and our paper is that they consider synchronous dynamics while the dynamics we consider are asynchronous.
One implication of being synchronous is that one might as well assume that all agents start by reporting their signals.  (Indeed, if all agents started null, they would switch to reporting their signals on the next step).
To illustrate the significance of this, consider the complete network as an example.
If agents all begin by declaring their reports then social learning will almost certainly occur, since the population will immediatley reach consensus on the majority opinion.  On the other hand, if agents begin with null reports and update asynchronously, then the entire population will copy the opinion of the first node that reports and hence there is a good chance that social learning does not occur.

Other lines of work in distributed computation focus on using properties of social networks to show that information can be aggregated efficiently in an algorithmic matter.  For example, Kempe Dobra and Gehrke \cite{KempeDG03} show that gossip-based protocols are particularly successful at aggregating information on networks with good expansion properties.

%% file: SNC-prelim.tex
\section{Model and Preliminaries}
\label{sec:prelim}

We consider a social network or graph $G = (V,E)$ with $|V| = n$ individuals.  Write $d(v)$ for the degree of $v$ in $G$, and $Vol(V) = \sum_{v \in V} d(v)$ for the volume of $V$ in $G$.  Individuals live in a world that is in one of two states, say red or blue.  Each individual $v$ has a private signal $X(v)\in\{red,blue\}$ regarding the state of the world. These $X(v)$ are conditionally indpendent given the state and are correct with probability $1/2+\delta$.  It will be convenient to assume, without loss of generality, that the state of the world is red and think of $red=1$ and $blue=0$.  Thus $\Pr[X(v)=1]=1/2+\delta$ for all $v$.

The individuals stochastically form and vocalize opinions about the state of the world. Let $C^t(v)\in\{red, blue, uncolored\}$ be the opinion of individual $v$ (or, equivalently, the color of node $v$) at time $t$.  Initially, individuals hold no opinions and so $C^0(v)=uncolored$.  Denote by $N_R^t(v)$ the number of $v$'s neighbors that are colored red at time $t$, and similarly denote $N_B^t(v)$ the number of $v$'s neighbors that are colored blue at time $t$. At every time $t > 0$, a node $v \in V$ is chosen uniformly at random. If $N_R^t(v) > N_B^t(v)$, then $v$ is colored red. If $N_R^t(v) < N_B^t(v)$, then $v$ is colored blue. If $N_R^t(v) = N_R^t(v)$, then $v$ is colored $X(v)$.

We first show that for a any graph $G$, this process stabilizes. That is, with probability $1$ there exists a $t < \infty$ such that $C^t(v) = C^{t'}(v)$ for all $t' \geq t$.
We do so in a standard way: define a potential function that is initially finite, bounded from below, and decreases by a constant amount in each time step.  Intuitively, our potential function counts a combination of the number of bichromatic edges in the graph and the number of self-disagreements, i.e., nodes whose stated opinion differs from their private signal. 

\begin{proposition}
For all $G$, with probability $1$, there exists a $t$ such that $C^t(v) = C^{t'}(v)$ for all $t' \geq t$. \smw{Furthermore, the expected number of steps until stabilization is at most $|V|^2 + 2|V||E|$. }
\end{proposition}

\begin{proof}
Define a potential function $F^t(v)$ that is $1$ if and only if $C^t(v) \neq X(v)$, and $0$ otherwise. Also define a potential function $G^t(e = (u,v))$ that is $2$ if either $u$ or $v$ is uncolored, or if $C^t(u) \neq C^t(v)$, and $0$ otherwise. Finally, define a potential function $H(t) = \sum_v F^t(v) + \sum_e G^t(e)$. Then $H(0) = |V| + 2|E|$. Furthermore, we claim that if any node's color is changed at time $t$, then $H(t) < H(t-1)$.

If a node $v$ is the first node in its neighborhood to change from uncolored to colored, then $F^t(v) < F^{t-1}(v)$. Furthermore, $G^{t-1}(e) = 2$ for all $e$ containing $v$ since $v$ was uncolored, so $G^t(e) \leq G^{t-1}(e)$ for all $e$, and $H(t) < H(t-1)$. If some nodes in $v$'s neighborhood were already colored, then $v$'s color is guaranteed to match the color of at least one neighbor and so for that edge $G^t(e)< G^{t-1}(e)$.  For all other edges $G^t(e)\leq G^{t-1}(e)$, and clearly $F^t(v)\leq F^{t-1}(v)$ and so $H(t)<H(t-1)$.

If a node changes colors, then maybe there was a tie among its neighbors. In this case, $\sum_e G^t(e) = \sum_e G^{t-1}(e)$, because we just switch the edges containing $v$ that disagree. But because the color changed with a tie, it must be the case that $F^{t-1}(v) = 1$ and $F^t(v) = 0$. So again $H(t) < H(t-1)$. Finally, maybe a node changed colors because of a majority among its neighbors. In this case, maybe $F^t(v) = F^{t-1}(v) + 1$, but $\sum_e G^t(e) \leq \sum_e G^{t-1}(e) - 2$ because at least one more edge switches from disagreement to agreement.

Thus, every time a node changes colors (or becomes colored for the first time), the value of $H$ decreases by at least $1$, and $H(0) = |V|+ 2|E|$, so the process stabilizes after at most $|V| + 2|E|$ changes. \smw{If the process has not already stabilized, then there is at least one node that would change colors (or becomes colored for the first time) and it is selected with probability $1/|V|$. So at every step independently there is a color change with probability at least $1/|V|$. Therefore the expected number of steps until a color change is bounded by $|V|$. As the total number of color changes is bounded by $|V| + 2|E|$, the expected number of steps until the process converges is at most $|V|^2 + 2|V||E|$.} 
\end{proof}

It is important to emphasize the distinction between {\em correct majority} and {\em consensus}. 
The former means that more than half of the nodes in the graph are stating the ``correct" opinion, while the latter means that every node in the graph is stating the same opinion (not necessarily the correct one).

We conclude this section with formal definitions of sparsity and expansiveness.

\begin{definition}(Sparsity) There are several different ways to state formally that a graph is sparse. In order from \smw{most restrictive to least restrictive}, this includes:
\begin{itemize}
\item Low fixed degree: The graph is $d$-regular, and $d$ is small.
\item Low maximum degree: Every node in the graph has degree at most $d$, and $d$ is small. 
\item Low arboricity: The graph is an edge-union of at most $d$ trees, and $d$ is small. 
\item Low average degree: The number of edges in the graph is at most $dn$, and $d$ is small. 
\end{itemize}
\end{definition}

Our main result considers the {\em maximum degree $d$} notion of sparsity.
The example in Section~5.2 of the full version shows that the low average degree notion of sparsity is not restrictive enough to guarantee a correct majority.
Our main open question asks whether or not our main result extends to low arboricity as well.

\begin{definition}(Weighted Adjacency Matrix)
The weighted adjacency matrix of a graph $G$, say $M = M(G)$, is an $n \times n$ matrix defined by
\[
M(x,y) =
	\begin{cases}
		\frac{1}{\sqrt{d(x)d(y)}}  & \text{if $x$ and $y$ are adjacent in $G$,}\\
		0 & \text{otherwise.}
	\end{cases}
\]
\end{definition}

\begin{definition}(Expansiveness)
A graph $G$ is a $\lambda$-expander if all but the first eigenvalue of the weighted adjacency matrix of $G$ lies in $[-\lambda,\lambda]$.
\end{definition}

%% file: SNC-examples.tex
\section{Examples}
\label{sec:examples}

To build intuition for our model and motivate our conjectures, we work through a few specific network topologies in detail before proving our main positive result.

\paragraph{Complete Graphs}
Suppose that $G$ is the complete graph on $n$ vertices.  The dynamics proceeds as follows: the node selected in round $1$, say $v_1$, will set $C^1(v_1) = X(v_1)$.  That is, $v_1$ reports its private signal.  Every subsequently chosen node will report the majority opinion of the population, and simple induction shows that this will be $X(v_1)$ at all times.  The process will therefore stabilize in a consensus on report $X(v_1)$ with probability $1$ for all $n$.  Since $\Pr[X(v_1) = 1] = 1/2 + \delta$, this consensus is correct with probability only $1/2 + \delta$.  In other words, the complete graph reaches consensus surely, but exhibits an extreme information cascade in which the population exhibits herding on the first reported signal.

\paragraph{Star Graphs}
We next show that under the star topology, the population will reach a correct consensus with high probability.  Suppose $G$ is a star with $n$ leaves.  First, we show that the population will certainly reach consensus on the first opinion reported by the center node, say $v$.

\begin{claim}
Suppose $v$ is selected by the dynamics for the first time in round $t_1$.  Then, with probability $1$, the dynamics reaches consensus on opinion $C^{t_1}(v)$.
\end{claim}
\begin{proof}
Suppose $C^{t_1}(v) = R$; the case $C^{t_1}(v) = B$ is handled identically.  Then $N_R^{t_1}(v) \geq N_B^{t_1}(v)$, with equality only if $X(v) = R$.  For any $t' > {t_1}$, if a leaf $u \neq v$ is chosen for update, then $C^{t'}(u) = C^{t'}(v)$.  That is, node $u$ will copy the opinion of $v$.  Simple induction then shows that, if we write $t_2 > {t_1}$ for the random variable indicating the round in which $v$ is selected for the second time, we must have $N_R^{t_2}(v) - N_B^{t_2}(v) \geq N_R^{t_1}(v) - N_B^{t_1}(v)$, and hence $C^{t_2}(v) = R$.  Applying this argument inductively, we conclude that $C^{t'}(v) = R$ for all $t' > {t_1}$.  Thus each leaf will adopt opinion $R$ each time it is selected for update after time $t$, and hence the population reaches consensus on $R$ with probability $1$.
\end{proof}

Write $t_1$ for the random variable representing the first report time of node $v$.  It remains to show that $C^{t_1}(v) = R$ with high probability.  By symmetry, the probability that $v$ chooses an opinion before at least $k$ leaves have chosen opinions is $k/(n+1)$.  Conditioning on the event that at least $k$ leaves have reported before $t_1$, each of their opinions matches their private signals.  Applying the additive Chernoff bound, the probability that at most half of them report $R$ at time $t_1$ is at most
\[ \Pr\left[ C_R^{t_1} \leq \left(\frac{1}{2} + \delta\right)k - \delta k \right] < e^{-2k\delta^2} \]
Choosing $k = \frac{1}{2\delta^2}\log(n)$ and taking a union bound, we conclude that the probability that at least $k$ leaves are selected before $v$, and that a majority of those selected leaves take opinion $R$, is at least $1 - \frac{\log(n)}{2\delta^2 n} - \frac{1}{n} = 1 - o(1)$.  We therefore conclude that with probability $1 - o(1)$ the star topology stabilizes in a correct majority. 

%% file: SNC-main.tex
\section{Majority and Consensus}
\label{sec:main}

In this section we 
give a sufficient condition for reaching a correct consensus.
More precisely, we focus on a family of $\lambda$-expanders of max-degree $d$ and prove that they converge to a correct consensus with high probability.

\begin{theorem}
\label{thm:main.expander}
Let $G$ be a $\lambda$-expander of max-degree $d$ with $\lambda\leq\delta/6$.  Then with probability at least $1 -  O(\frac{1}{(\delta \ln \ln n)^2})$, the process will terminate in a red consensus.
%
\end{theorem}


Here is a brief outline of our proof. First, we show that in any graph with max-degree $d$ (not necessarily an expander), the volume of nodes with opinion red after $O(n/\delta)$ steps of the process is at least $(1/2 + \delta/2)|E|$ with high probability. We do this by showing that the expected volume of currently red nodes is at least $(1/2 + \delta)|E|$, and then bounding the total pairwise correlation among the colors of nodes to be $o(|E|)$. Combining these two facts with Chebyshev's inequality gives us the desired claim. Next, we show that for all sufficiently expansive graphs, continuing the stochastic process from a point when the volume of red nodes is at least $(1/2 + \delta/2)|E|$ nodes will result in a red consensus with high probability. \smw{Formally, the proof of Theorem~\ref{thm:main.expander} follows from Proposition~\ref{prop:make this a proposition} and Corollary~\ref{laststep} after observing that the probability in Corollary~\ref{laststep} is asymptotically dominated by that in Proposition~\ref{prop:make this a proposition}.}

\subsection{Low Degree and Correctness}

We would like to count the expected volume of red nodes after a linear number of steps.  To this end, we define a Boolean function that specifies the color of a node after a finite sequence of updates.  Specifically, let $S$ be any finite sequence of nodes and define a Boolean function $f^S_v$ that takes as input the private signals $\mathcal{X}=\{X(u)\ |\ \forall u \in V\}$ and outputs the color $C^{|S|}(v)$, when the process chooses nodes in the order specified by $S$ and the private signals are $\mathcal{X}$.  If $C^{|S|}(v)$ is uncolored, we define $f^S_v(\mathcal{X})$ to output the private signal $X(v)$; we will later show that this induces a limited degree of overcounting as most nodes are colored after a linear number of steps.  Define a random variable $$f^S(\mathcal{X})=\sum_v d(v) f^S_v(\mathcal{X})$$ that counts the volume of red nodes after sequence $S$.
Now fix a sequence length $T$ and
%
%
let $f_T$ be the random variable that selects a sequence $S$ of length $T$ and signals $\mathcal{X}$ at random and outputs $f^{S}(\mathcal{X})$.
Then $f_T$ is the volume of red nodes after $T$ steps of our process.
We bound the expectation and variance of $f_T$ and apply Chebyshev to prove that the volume of red nodes is a majority with high probability.

\subsubsection{Bounding the expectation}
To bound the expectation, note each $f^S_v$ is \emph{monotone} for all $S, v$. That is, switching any set of input signals from blue to red can only cause $f^S_v$ to switch from blue to red, but not from red to blue. In addition, $f^S_v$ is \emph{odd} for all $S,v$. That is, switching all input signals from blue to red and red to blue will cause the output to flip. The following theorem due to Mossel, Neeman, and Tamuz \cite{Mossel2013}, which uses Boolean function analysis, states that such functions on biased random inputs have biased outputs.

\begin{theorem}\label{thm:oddmonotone} (\cite{Mossel2013})
Let $f$ be an odd, monotone Boolean function. Let $X_1,\ldots,X_n$ be input bits, each sampled i.i.d.\ from a distribution that is $1$ with probability $p \geq 1/2$ and $0$ otherwise. Then $\mathbb{E}[f(X_1,\ldots,X_n)] \geq p$.
\end{theorem}

The following corollary is a direct application of Theorem~\ref{thm:oddmonotone} and the fact that the private signals $X_i(v)$ are red with probability at least $1/2+\delta$.

\begin{corollary}\label{cor:expected}
The expected volume of red nodes at time $T$, for any $T$, is at least $(1/2 + \delta)|E|$. That is, $\mathbb{E}[f_T] \geq (1/2 +\delta)|E|$ for all $T$.
\end{corollary}

\subsubsection{Bounding the variance}
In light of Corollary~\ref{cor:expected}, if we can also bound the variance of $f_T$, then we can use Chebyshev's inequality to argue that $f_T \geq (1/2 + \delta/2)|E|$ with high probability. Formally, let's define the $f^S$ so that there are $n^T$ separate copies of $G$, and the private signals $X(v)$ are sampled independently for each copy. Then let $f_T$ be the random variable that picks one $S$ and its corresponding $G$ uniformly at random and outputs $f^S$. We first state a lemma that allows us to analyze the variance of $f_T$.

\begin{lemma}\label{lem:variance}
Let $\{X_1,\ldots,X_n\}$ be random variables all with the same expectation $\mathbb{E}[X_i] = c$, and let $X$ be a random variable that samples from $\{X_1,\ldots,X_n\}$ uniformly at random. Then $Var(X) = \frac{1}{n} \sum_i Var(X_i)$.
\end{lemma}

\begin{proof}
$Var(X) = \mathbb{E}[X^2] - \mathbb{E}[X]^2 = \mathbb{E}[X^2] - c^2$. $\mathbb{E}[X^2] = \frac{1}{n} \sum_i \mathbb{E}[X_i^2]$. So we get:

\begin{eqnarray*}
Var(X) & = & \frac{1}{n} \sum_i \mathbb{E}[X_i^2] - c^2 \\
& = & \frac{1}{n} \sum_i \mathbb{E}[X_i^2] - \mathbb{E}[X_i]^2  \\
& = & \frac{1}{n} \sum_i Var(X_i).
\end{eqnarray*}
\end{proof}

To use Lemma~\ref{lem:variance}, we need to modify our random variables slightly so that they all have the same expectation. To do this, just define $g^S = f^S - (\mathbb{E}[f^S] - \frac{1+2\delta}{2}|E|)$, and $g_T$ to sample $S$ uniformly at random and then sample $g^S$. By Corollary~\ref{cor:expected}, $f^S \geq g^S$ for all $S$ always. Therefore, showing that $g_T \geq (1/2 + \delta/2)|E|$ with high probability suffices to prove that $f_T \geq (1/2 + \delta/2)|E|$ as well.

So now let's analyze the variance of $g_T$. Lemma~\ref{lem:variance} tells us that the variance of $g_T$ is just the average of the variances of each $Var(g^S)$. Furthermore, we can write the variance of each $g^S$ as
\[ Var(g^S) = \sum_{u,v} d(u)d(v)Cov(f^S_u, f^S_v) \]
and therefore, we can write $Var(g_T)$ as
\[ Var(g_T) = \frac{1}{n^T} \sum_S \sum_{u,v} d(u)d(v)Cov(f^S_u,f^S_v). \]

Now we observe that $Var(g_T)$ is exactly the expected value of the following random process: sample two nodes $u$ and $v$ uniformly at random (with replacement), sample a sequence of length $T$ uniformly at random, and compute $n^2d(u)d(v)Cov(f^S_u,f^S_v)$. Furthermore, as each $f^S_v$ is a 0-1 random variable, $Cov(f^S_u,f^S_v) \leq 1$. As $Cov(f^S_u,f^S_v) = 0$ when $f^S_u$ and $f^S_v$ are independent, we can define $G_T$ to be a random variable that is $0$ whenever $S,u,v$ are sampled such that $f^S_u$ and $f^S_v$ are independent and $1$ otherwise. The reasoning above shows that if we show that $\mathbb{E}[G_T] \leq c$, then $Var(g_T) \leq cd^2n^2$.

So now our aim is to study $G_T$. Let's first ask what private signals can possibly affect the color of node $v$ at the end of sequence $S$. If $t_v$ is the last step that $v$ is chosen to update its color, then $f^S_v$ is clearly a function of the colors of $v$'s neighbors at time $t_v$. Furthermore, if we look at any neighbor $u$ of $v$, and let $t_u$ be the last step that $u$ is chosen to update its color before $t_v$, then the color of $u$ at time $t$, is clearly a function of the colors of $u$'s neighbors at time $t_u$ \smw{(as $t \geq t_u$, and node $u$ does not update its color between $t$ and $t_u$)}. Iterating this reasoning out, we can define the set $N^S(v)$ to be those nodes $u$ such that there is a path $v,x_1,\ldots,x_k,u$ from $u$ to $v$ and corresponding times $t_v> t_1 >\ldots > t_k > t_u$ such that $u$ announces its color at time $t_u$ in $S$, $v$ announces its color at time $t_v$ in $S$, and each $x_i$ announces its color at time $t_i$ in $S$. We then see that $f^S_v$ can be written as a function of only the signals $\{X(u)\}_{u \in N^S(v)}$. Therefore, if $N^S(v) \cap N^S(u) = \emptyset$, it is necessarily the case that $f^S_u$ and $f^S_v$ are independent, as they are functions on disjoint sets of independent random variables. So our approach to bounding $\mathbb{E}[G_T]$ will be to analyze the probability that when $v$ and $u$ are chosen uniformly at random (with replacement) and $S$ is a random sequence of length $T$ that $N^S(v) \cap N^S(u) = \emptyset$.

We do this by studying the random variable $|N^S(v)|$ for a random node $v$ and random sequence $S$. We can compute $N^S(v)$ by initializing $N^S(v) = \emptyset$ and tracking backwards through $S$. Until the first (moving backwards in time) time that $v$ announces its color, $N^S(v) = \emptyset$. When $v$ first updates its color, we update $N^S(v) = \{v\}$. From here, until the next time that a neighbor of $v$ announces its color, $N^S(v)$ remains unchanged. When the first neighbor $u$ of $v$ updates its color, we update $N^S(v) = \{v,u\}$. Iterating this reasoning, we can compute $N^S(v)$ by tracking backwards through $S$, updating $N^S(v)$ to $\{v\}$ the first time that $v$ announces its color, and then updating $N^S(v) := N^S(v) \cup \{u\}$ any time a neighbor $u$ of $N^S(v)$ announces its color.

So let $N_i$ be the random variable denoting the number of steps between when $|N^S(v)|$ first becomes $i-1$ and when $|N^S(v)|$ first becomes $i$ over the random choice of $S$. Recall $S$ is chosen uniformly at random from all sequences of length $T$.  As each node has degree at most $d$, and the neighborhood $N^S(v)$ is a connected subgraph, when $|N^S(v)|=i-1\geq2$, there are at most $(i-1)(d-1)$ ways to grow $N^S(v)$ (and for $i-1=1$, there are at most $d$ ways).  Thus the $N_i$ are independent geometric random variables with mean at least $\frac{n}{1+(i-1)(d-1)}$.  For ease of analysis, we analyze each $N_i$ as independent random variables of mean exactly $\frac{n}{id}$ (this is valid because these random variables are stochastically dominated by the actual $N_i$, meaning that we are only underestimating the number of steps needed for $|N^S(v)|$ to grow). Now we see that, for any $x$, if we define $N^x$ to be the number of steps before $|N^S(v)| = x$, then $N^x$ is exactly $\sum_{i=1}^x N_i$. As each $N_i$ is a geometric random variable with parameter $id/n$, $\mathbb{E}[N_i] =\frac{n}{di}$, and $Var(N_i)= \frac{n^2}{i^2d^2}$. So because all $N_i$ are independent, we get that:
$$\mathbb{E}[N^x] = \sum_{i=1}^x \frac{n}{di} \geq \frac{n\ln x}{d}, \quad
Var(N^x) = \sum_{i=1}^x \frac{n^2}{i^2d^2} = \frac{\pi^2n^2}{6d^2},$$
$$\sigma(N^x) = \sqrt{Var(N^x)} \leq \frac{2n}{d}$$
So by Chebyshev's inequality, we get that
$Pr[N^x \leq \frac{n\ln x}{d} - t\frac{2n}{d}] \leq \frac{1}{t^2}$,
which can be rewritten as:
\begin{equation}\label{eq:prob1}
Pr\left[N^x \leq (1-\epsilon) \frac{n\ln x}{d}\right] \leq \frac{4}{(\epsilon \ln x)^2}
\end{equation}
From here, we simply observe that if the shortest path from $u$ to $v$ has length $ > 2x$, and $|N^S(v)|, |N^S(u)| \leq x$, then $N^S(u) \cap N^S(v) = \emptyset$. We also observe that the number of nodes within distance $2x$ of $v$ is bounded by $d^{2x}$ for all $x$. So when $u$ and $v$ are chosen uniformly at random (with replacement) we have:
$$Pr[dist(u,v) \leq 2x] \leq \frac{d^{2x}}{n}$$
Taking $T = \frac{n\ln x}{2d}$ \smw{corresponds to setting $\epsilon = 1/2$ in Equation~\eqref{eq:prob1}. So for any $u,v$, the union bound guarantees that with probability at most $\frac{32}{(\ln x)^2}$ $|N^S(v)|, |N^S(u)| \geq x$. Furthermore, if $u,v, S$ are chosen uniformly at random, we see that with probability at most $\frac{d^{2x}}{n}$, $dist(u,v) \leq 2x$. Again taking a union bound,} the probability that either of these events occur is at most $\frac{32}{(\ln x)^2} + \frac{d^{2x}}{n}$. And in the event that none of these events occur, we clearly have $N^S(u) \cap N^S(v) = \emptyset$. Therefore, we conclude that for all $x$, if $T = \frac{n\ln x}{2d}$, $\mathbb{E}[G_T] \leq \frac{32}{(\ln x)^2} + \frac{d^{2x}}{n}$.

By the reasoning above, we have now shown that when $T = \frac{n\ln x}{2d}$, we have:
$$Var(g_T) \leq d^2n^2\left(\frac{32}{(\ln x)^2} + \frac{d^{2x}}{n}\right)$$
To simplify notation, we observe that whenever $x = o(\log n)$ the first term asymptotically dominates the second. So we will restrict ourselves to setting $x = o(\log n)$ and rewrite:
$$Var(g_T) \leq \frac{33d^2n^2}{(\ln x)^2}$$
So we can apply Chebyshev's inequality to $g_T$ now and see that whenever $x = o(\log n)$, we have:
$$Pr\left[g_T \leq (1/2 + \delta)|E| - t \cdot \frac{8dn}{\ln x}\right] \leq \frac{1}{t^2}$$
And plugging in for $t = \delta \ln(x)/(32d)$ we get:
$$Pr[g_T \leq (1/2 + 3\delta/4)|E|] \leq \frac{1024d^2}{(\delta \ln x)^2}$$
And because $f_T \geq g_T$ always, we have:
$$Pr[f_T \leq (1/2 + 3\delta/4)|E|] \leq \frac {1024d^2}{(\delta \ln x)^2}$$
Finally, recall that in order to make $f_T$ odd, we had to define $f^S_v$ to be $X(v)$ in the event that $v$ does not announce its color at all in $S$. So $f_T$ does not exactly count the number of red nodes because its getting credit for some nodes with a red private signal who haven't actually announced a color at all. But this is easy to cope with: we can just show that with high probability the volume of nodes that have yet to announce a color after $\frac{n\ln x}{2d}$ steps is at most $\delta |E|/4$. Note that because all nodes have degree at most $d$, it is sufficient to show that the \emph{number} of nodes who have yet to announce a color is at most $\delta n/(4d) \leq \delta |E|/(4d)$ with high probability.

For a single node $v$, the probability that $v$ has not yet announced a color after $\frac{n\ln x}{2d}$ is exactly:
$$(1-1/n)^{\frac{n\ln x}{2d}}\leq e^{-\frac{\ln x}{2d}} \leq x^{-\frac{1}{2d}}$$
So if we define $C_x(v)$ to be the indicator random variable that is $1$ if $v$ has not yet announced a color by time $\frac{n \ln x}{2d}$, and $0$ otherwise, the collection of random variables $\{C_x(v)\}_v$ are negatively correlated. So if we define $C_x = \sum_v C_x(v)$, we get
$\mathbb{E}[C_x] = nx^{-\frac{1}{2d}}$.
Using the additive Chernoff bound, we get:
$$Pr\left[C_x \geq nx^{-\frac{1}{2d}} + tn\right] \leq e^{-2t^2n}$$
And plugging in for $t = \delta/(4d) - x^{-\frac{1}{2d}}$ we get:
$$Pr[C_x \geq \delta n/(4d)] \leq e^{-n(\delta/(4d) - x^{-\frac{1}{2d}})^2}$$
Because $\delta$ and $d$ are constant and $x = o(\log n)$,\footnote{In fact, this would still be true if we took $x = O((\log n)^{1-\epsilon})$ for some $\epsilon > 0$, $1/\delta = O(x)$, and $d = o(\frac{\ln x}{\ln (1/\delta)})$} this is clearly asymptotically dominated by $\frac{1}{(\delta \ln x)^2}$. So taking a union bound over the probability that more than $\delta/4$ nodes have yet to announce a color and the probability that $f_T \leq (1/2 + 3\delta/4)n$, we get the following proposition:

\begin{proposition}\label{prop:make this a proposition}
For any $x = o(\log n)$ and $T = \frac{n\ln x}{2d}$:
\begin{equation*}
\begin{split}
Pr[\mbox{volume\ of\ \ announced\ reds\ at\ time\ } T \leq (1/2 + \delta/2)|E|] \leq \frac{1025}{(\delta \ln x)^2}
\end{split}
\end{equation*}
In particular, when $x = \ln\ln n$ and $T = \frac{n \ln \ln \ln n}{2d}$, this probability is at most $O\left(\frac{1}{(\delta \ln \ln n)^2}\right)$
\end{proposition}

\subsection{Expansion and Consensus}
In this section, we apply a different argument based on expansion to show that if $G$ is sufficiently expansive, once the volume of nodes that have announced red exceeds $(1/2 + \delta/2)|E|$, it is extremely likely that the process will continue to stabilize in a red consensus. This argument has two steps. First, we apply an argument of~\cite{Mossel2013} to show that, in an expansive network, the volume of nodes that will switch from blue to red if chosen is a constant factor larger than those that would switch from red to blue if chosen, conditioned on the fact that the volume of nodes announcing red is at least $(1/2 + \delta/4)|E|$.  
Second, we argue that with very high probability, due to this fact, if the volume of nodes announcing red starts above $(1/2 + \delta/2)|E|$, then we will reach the point where all nodes have announced red before we reach a point where the volume of nodes announcing red is only $(1/2 + \delta/4)|E|$.  This second step proceeds by coupling the convergence process to an absorbing random walk, and applying the theory of biased random walks.

In the following lemmas, let $R$ denote the set of nodes who have currently announced red, and $B$ the set of nodes who have currently announced blue or nothing. Let also $R'$ denote the set of nodes that would announce red if they were chosen, and $B'$ the set of nodes that would announce blue if they were chosen.

The following lemma relates the number of edges between two sets of nodes in an expander with max-degree $d$ to their expected number in a random graph.

\begin{lemma}\label{lem:mixing}(\cite{ChungG08})
If $G$ is a $\lambda$-expander of max-degree $d$, then for any two subsets $S,T \subseteq V$, let $E(S,T)$ denote the number of edges between $S$ and $T$ (double-counting edges from $S \cap T$ to itself). Then:
$$|E(S,T) - \frac{Vol(S)Vol(T)}{|E|}| \leq \lambda \sqrt{Vol(S)Vol(T)}$$
\end{lemma}

Using Lemma \ref{lem:mixing}, we can bound the number of ``potential" B nodes. 

\begin{corollary}\label{cor:allrange} If $G$ is a $\lambda$-expander of max-degree $d$ with $\lambda \leq \frac{\delta}{6}$ and $|R| \geq (1/2 + \delta/4)n$, then $|B'| \leq |B|/2$.
\end{corollary}

\begin{proof}
We know that every node in $B'$ has at least half of its neighbors in $B$ (or else they would choose red). Therefore, $E(B',B) \geq Vol(B')/2$. In addition, Lemma \ref{lem:mixing} tells us that $E(B',B) \leq \frac{Vol(B')Vol(B)}{|E|} + \lambda \sqrt{Vol(B)Vol(B')}$.  Putting these two together, we get:
%
%
$$Vol(B')/2 \leq  \frac{Vol(B')Vol(B)}{|E|} + \lambda \sqrt{Vol(B')Vol(B)}.$$
Reorganizing the last inequality we get
$$Vol(B') \leq Vol(B)  \left(\frac{\lambda}{\frac{1}{2} - \frac{Vol(B)}{|E|}}\right)^2.$$
Applying the fact that $Vol(B)/|E| \leq 1/2 - \delta/4$ we get
$$Vol(B') \leq Vol(B)\left(\frac{16\lambda^2}{ \delta^2}\right).$$
Finally, by the fact that $\lambda \leq \frac{\delta}{6}$ we get
$$Vol(B') \leq Vol(B) \left(\frac{16\delta^2}{36\delta^2}\right) \leq Vol(B)/2,$$
as desired.
\end{proof}

Now, we make use of Corollary~\ref{cor:allrange} to show that we are very likely to switch more blues to reds than reds to blues over many announcements. 

\begin{corollary}\label{cor:bluevred}
If $Vol(B') \leq Vol(B)/c$, then $Vol(B\cap R') \geq c Vol(R \cap B')$, and $Vol(B \cap R') \geq 1$. In other words, the volume of nodes who will switch from blue to red is at least $c$ times the number of nodes who will switch from red to blue if chosen, and there is at least $1$ such node.
\end{corollary}

\begin{proof}
We know that $Vol(B \cap B') = x$, for some $x \geq 0$. So we can write $Vol(B \cap R') = Vol(B) - x$ and $Vol(B' \cap R) = Vol(B') - x$. Combining this with the fact that $Vol(B) \geq cVol(B')$ we get:
$$\frac{Vol(B \cap R')}{Vol(R \cap B')} \geq \frac{cVol(B') - x}{Vol(B')-x}$$
Because $x \geq 0$, this is always at least $c$.
As $Vol(B') < Vol(B)$, there must be at least one node in $B \cap R'$.
\end{proof}

To complete our analysis, we use the theory of biased random walks. 

\begin{definition}
For $d \geq 1$ and $p > 0$, a $d$-bounded, $p$-biased random walk on the integers is a sequence $(Z_t)_{t \geq 0}$ such that:
\begin{itemize}
\item $Z_0 = 0$,
\item $Z_t$ depends only on $(Z_0, \dotsc, Z_{t-1})$,
\item $|Z_t - Z_{t-1}| \leq d$ for each $t \geq 1$, and
\item for all $(Z_t)_{t < T}$, $\E[Z_T \ |\ Z_0, \dotsc, Z_{T-1}] \geq Z_{T-1} + p$.
\end{itemize}
\end{definition}

The following lemma establishes a crucial property of biased random walks, which is then used in the remainder of this section to show that once the volume of red nodes reaches a certain threshold, the process will converge to a red consensus with high probability.

\begin{lemma}
\label{lem:martingale}
Let $(Z_t)_{t \geq 0}$ be a $d$-bounded $p$-biased random walk on the integers. Then, for any $x > 0$, the probability that the walk reaches a value less than $-x$ before a value greater than $x$ is at most $\frac{2x}{p} e^{-px/4d^2}$.
\end{lemma}

\begin{proof}
For each $t \geq 1$, define $Y_t = Z_t - Z_{t-1}$, and let $W_t = Y_t - \E[Y_t\ |\ Y_1, \dotsc, Y_{t-1}]$.  Note that the sequence $(W_t)_{t \geq 1}$ forms a martingale, whose entries lie in $[-d, d]$.  The Azuma-Hoeffding inequality then implies that, for any $n \geq 1$,
\[ \Pr\left[ \sum_{t = 1}^n W_t < - x \right] \leq e^{-x^2 / 2nd^2}. \]
Let $A_n$ be the event that there exists any prefix of the sequence $(W_t)_{t \leq n}$ with sum less than $-x$.
Taking a union bound over all $t$ between $1$ and $n$, we have that the probability of event $A_n$ occurring is at most $n \cdot e^{-x^2 / 2nd^2}$.

If we condition on $A_n$ not occurring, then observe that for each $T \leq n$,
\begin{align*}
Z_T & = \sum_{t = 1}^{T} Y_t = \sum_{t = 1}^{T} W_t + \E[Y_t\ |\ (Y_k)_{k < t}] > \E[Z_T] - x.
\end{align*}
In particular, $Z_n > \E[Z_n] - x$ and moreover $Z_t > -x$ for all $t \leq n$.  If we choose $n = 2x/p$, then $\E[Z_n] > pn = 2x$, and hence $A_{2x/p}$ not occurring implies that $Z_n > x$ and $Z_t > -x$ for all $t < n$, as required.  Furthermore, the probability of $A_{2x/p}$ is at most $\frac{2x}{p} \cdot e^{-px / 4d^2 }$.
\end{proof}

\smw{We now apply Lemma~\ref{lem:martingale} to the stochastic process, letting $Z_t$ be the volume of red nodes. The hypotheses of Corollary~\ref{cor:randomwalk} below (and the fact that $G$ has maximum degree $d$) guarantee that the random walk is $\frac{c-1}{c+1}$-biased and $d$-bounded.}

\begin{corollary}\label{cor:randomwalk}
Let $R_0$ and $B_0$ be such that $Vol(R' \cap B_0) \geq c Vol(B' \cap R_0)$. For any $x$, if $R$ and $B$ maintain this property whenever $Vol(B_0) - x \leq Vol(B) \leq Vol(B_0) + x$ (and therefore $Vol(R_0) - x \leq Vol(R) \leq Vol(R_0) + x$ as well), then the probability that we arrive at a state with $Vol(B)\geq Vol(B_0) + x$ before one with $Vol(R)\geq Vol(R_0)+x$ is at most $2x(\frac{c+1}{c-1}) e^{-(c-1)x / 4(c+1)d^2}$.
\end{corollary}

\begin{proof} 
Consider a biased one-dimensional random walk that takes $\ell$ steps up whenever a node of degree $\ell$ switches from blue to red, and $\ell$ steps down whenever a node of degree $\ell$ switches from red to blue. Then the corollary is exactly studying the probability that this random walk reaches a depth of $-x$ before a height of $x$.

This walk is $d$-bounded.  We also claim that it is $(\frac{c-1}{c+1})$-biased.  To see this, let $W^+$ be the expected upward step of the walk on a given round; i.e., the expected step of the walk if we were to replace any negative movement by $0$.  Likewise, let $W^- \leq 0$ be the expected downward step.  Note then that the expected step is $W^+ + W^-$.  Since $Vol(R' \cap B_0) \geq c Vol(B' \cap R_0)$, we have $W^+ \geq c W^-$.  Also, $W^+ - W^- \geq 1$, since each step is of distance at least $1$.  We can then conclude that $W^+ + W^- \geq \frac{c-1}{c+1}(W^+ - W^-) = \frac{c-1}{c+1}$. Now, by Lemma~\ref{lem:martingale}, the probability that this walk reaches depth $-x$ first is at most $2x(\frac{c+1}{c-1})e^{-(c-1)x/4(c+1)d^2}$.
\end{proof}

\smw{Finally, we use Corollary~\ref{cor:randomwalk} to prove that the stochastic process terminates in a consensus. The idea is that once we have reached $Vol(R) \geq (1/2 + \delta/2)|E|$, the expansiveness of $G$ guarantees that the hypotheses of Corollary~\ref{cor:randomwalk} are satisfied. We then iteratively apply Corollary~\ref{cor:randomwalk} to show that we are extremely likely to reach a state with $Vol(R) \geq (1/2 + k\delta/2)|E|$ before we reach a state with $Vol(R) \leq |E|/2$, for all integers $k \in [2/\delta]$.}

\begin{corollary}\label{laststep}
If $G$ is a $\lambda$-expander with max-degree $d$ and with $\lambda \leq \frac{\delta}{6}$, and the stochastic process reaches a point where $Vol(R) \geq (1/2 + \delta/2)|E|$, then with probability at least $1 - 4 n \cdot e^{-\delta n/48d^2}$, the process will terminate in a red consensus.
\end{corollary}

\begin{proof}
Once the process reaches a point where $Vol(R) \geq (1/2 + \delta/2)|E|$, we will have $Vol(R) \geq (1/2 + \delta/4)|E|$ until the volume of reds that switch to blue is at least $\delta |E|/4$ more than the volume of blues that switch to red. Therefore, by Corollaries~\ref{cor:allrange},~\ref{cor:bluevred}, and~\ref{cor:randomwalk}, the probability that we reach a point where $Vol(R) = (1/2 + \delta/4)|E|$ before we reach a point where $Vol(R) = (1/2 + 3\delta/4)|E|$ is at most
\[ 6(\delta n / 4) e^{-(\delta n/4)/12d^2} < 2\delta n e^{-\delta n/48d^2}.\]
Similarly, once we have reached a point where $Vol(R) = (1/2 + i\delta/4)|E|$ ($2 \leq i < 2/\delta$), the probability that we reach a point where $Vol(R) = (1/2 + (i-1)\delta/4)|E|$ before we reach a point where $Vol(R) = (1/2 + (i+1)\delta/4)|E|$ is at most $2\delta n e^{-\delta n/48d^2}$. Therefore, we can take a union bound over all $2 \leq i < 2/\delta$ and say that with probability at least $1 - 4 n e^{-\delta n/48d^2}$, the first time we hit $Vol(R) = (1/2 + i\delta/4)|E|$, we will hit $Vol(R)= (1/2 + (i+1)\delta/4)|E|$ before we hit $Vol(R) = (1/2 + (i-1)\delta/4)|E|$, for all $2 \leq i < 2/\delta$. In the event that this happens, we will hit a red consensus before we hit $Vol(R) = (1/2 + \delta/4)|E|$, and therefore the process will stabilize in a red consensus.
\end{proof}

%% file: SNC-conclusion.tex
\section{Conclusion}
We study whether information aggregates efficiently under natural dynamics in social networks with ``real-world'' properties. We show that if each individual's signal agrees with the ground truth with probability at least $1/2 + \delta$, independently, then the entire society is likely to agree on the ground truth with high probability (approaching $1$ as $n \rightarrow \infty$) in the class of $\lambda$-expanders with maximum degree $d$ for any fixed $d,\lambda \leq \frac{\delta }{6}$.  We also analyze separately the example of a star on $n$ nodes, and show that it also achieves a consensus on the ground truth with high probability. This suggests that our results apply to additional notions of sparsity. An interesting direction for future work would be to show that more general classes of ``sparse'' expanders reach consensus on the ground truth with high probability. One possibility is the set of expanders with arboricity of at most $d$.
Additionally, the use of sparsity and expansiveness is decoupled in our analysis: sparsity is used to show that a correct majority is reached at some point during the process, and expansiveness is used to show that, once this occurred, the process terminates in a correct consensus.
These results suggest two interesting directions for future research.
First, we conjecture that sparsity (e.g., low arboricity) guarantees that
the process stabilizes in a correct majority, as in the ring.
Second, we showed that expansiveness guarantees that once enough of a (possibly incorrect) majority forms, the process terminates in a
consensus with high probability. We conjecture that all expansive graphs terminate in a (possibly incorrect) consensus with high probability.

%% file: SNC-hardexamples.tex
\section{Additional Examples}
\label{sec:examples2}

We now consider two examples of graph classes not covered by Theorem \ref{thm:main.expander}, motivating our conjecture that the sparsity condition of Theorem \ref{thm:main.expander} can be relaxed to include all graphs of bounded arboricity.

\subsection{Cycle Graphs}
\label{sec:cycles}

We first show that for the cycle topology, the population will reach a correct majority that is not a consensus, with high probability.  Suppose $G$ is a cycle on $n$ vertices.  We note the following straightforward fact.

\begin{claim}
\label{claim.ring}
Suppose $v_1, v_2, v_3, v_4$ is a path of length $4$ in $G$, such that $X(v_2) = X(v_3) = x$ and $v_2, v_3$ are selected for the first time before $v_1$ or $v_4$ are selected for the first time.  Then $C^t(v_2) = X(v_2)$ for all $t$ after $v_2$ is first selected, and similarly for $v_3$.
\end{claim}
\begin{proof}
Whichever of $v_2$, $v_3$ is selected first will choose its private signal, and the other will copy that private signal.  Then regardless of the reports of $v_1$, $v_4$ at any later time $t$, at least half of the neighborhood of $v_2$ will report $x$ at time $t$ (i.e., $v_3$), and hence it will select $x$ at time $t$.  A similar argument holds for $v_3$.
\end{proof}

We refer to a pair of neighboring nodes that satisfy the conditions of the above claim as a \emph{blocking pair}.  Each neighboring pair is a blocking pair with probability at least $\frac{1}{12}$, and moreover for $\delta < \frac{1}{2}$ the pair has opinion $R$ with constant probability and opinion $B$ with constant probability.  Thus, with high probability, the graph will contain two blocking pairs, one with opinion $R$ and the other with opinion $B$.  In this event, consensus does not occur.  We conclude that the population reaches consensus with vanishingly small probability.

Consider a sequence of $\log^2(n)$ consecutive paths of length $4$.  Each contains a blocking pair with probability at least $\frac{1}{12}$, independently, and hence the probability that this sequence has no blocking pairs is at most $e^{\frac{1}{12}\log^2(n)}$.  A union bound over all such sequences yields that with high probability the distance between any two consecutive blocking pairs is at most $\log^2(n)$.

There are therefore at least $n / \log^2(n)$ contiguous segments of $G$, defined as the paths between blocking pairs.  Furthermore, from the definition of a blocking pair, the final state of one such segment is independent of the private signals of nodes in other segments, conditioned on the opinions of the bordering blocking pairs.  We can therefore think of the graph as consisting of $n / \log^2(n)$ independent paths.  We refer to such a paths as an $RR$, $RB$, or $BB$ segment, referring to the colors selected by the blocking pairs on its boundary.

We now partition the nodes of the cycle into three groups.  The first group is the set of all nodes in blocking pairs.  With high probability, more nodes in this group converge to opinion $R$ than opinion $B$.

The second group is the set of all nodes in $RR$ or $BB$ segments, of length $1$ or $2$.  With high probability, there are $O(n)$ nodes in such segments, and the likelihood that such a segment is $RR$ (rather than $BB$) is at least $\frac{1}{2}+\delta$.  An $RR$ segment of length $1$ or $2$ must certainly converge to consensus on $R$, and similarly a $BB$ sigment of length $1$ or $2$ must converge to consensus on $B$.  Thus, with high probability, the difference between the number of nodes converging to $R$ versus the number converging to $B$ in such segments is $\Omega(\delta n)$.

The third group is the set of nodes in all other types of segments; let $N$ denote its size. With high probability, $N \geq cn$ for a fixed constant $c$, since only a small constant fraction of segments are length $1$ or $2$.  Let $N_R^\delta$ denote the number of such nodes that ultimately converge to opinion $R$, given our value of $\delta$.  Then note that, conditioning on the value of $N$, $N_R$ stochastically dominates $N_R^0$.  By symmetry, $\mathbb{E}[N_R^0] = N/2$.  Moreover, this number is a sum of independent random variables: the number of nodes in each segment that converge to opinion $R$.  Each of these random variables takes a value in $\{0, 1, \dotsc, \log^2(n)\}$.  The Hoeffding bound therefore implies that, for any quantity $y$,
\[ \Pr[ N_R^0 < N/2 - y ] \leq e^{- \frac{2 y^2}{cn (\log^2(n))^2 } }. \]
Taking $y = \log^5(n)\sqrt{n}$, we conclude that with probability at least $1 - 1/n$, $N_R^0$ will be at least $N/2 - y$.  We therefore have that $N_R \geq N/2 - y$ with probability at least $1 - 1/n$.  

Combining the three cases, we conclude that, with high probability, the number of nodes converging to $R$ is at least $n/2 + \Omega(\delta n) - O(\log^5(n) \sqrt{n}) = n/2 + \Omega( \delta n )$.  Thus, for $n$ sufficiently large, the majority of nodes converge to $R$ with high probability.

\subsection{Low Average Degree does not imply Correct Majority}
\label{sec:low-average}

We now show that there is a real difference between average degree and max degree and how it affects achieving a correct majority. Specifically, intuition suggests (and Proposition~\ref{prop:make this a proposition} confirms) that the stochastic process should reach (but not necessarily stabilize in) a correct majority after not too many iterations because each node's report within this timeframe should be ``basically'' independent of the others. We show that this intuition holds only when we define sparse graphs to be those with low maximum degree and not those with low average degree, as Proposition~\ref{prop:make this a proposition} fails to hold on a class of graphs with low average degree.

Define $G_m^\ell$ ($m \geq \ell$) to have nodes partitioned into two sets, $M$ and $L$. $M$ has $m$ nodes and forms a clique. $L$ can be partitioned into $L_1 \sqcup \ldots \sqcup  L_{m}$, where each $L_i$ has $\ell$ nodes, each adjacent only to node $i \in M$. In other words, $G_m^\ell$ has a clique $M$ of $m$ nodes in the middle, and each node of $M$ has $\ell$ leaves hanging off of it. We show essentially that, for a sufficiently large ratio of $m/\ell$, the first node in $M$ to be chosen is likely to announce his own signal, and then every subsequent node in $M$ to be chosen will just copy that signal. At a high level, here is what happens: if every node in $M$ chosen so far has copied the first node's signal, and if, whenever a new node in $M$ is chosen, more nodes in $M$ have been chosen than its leaves, this new node will also copy the other nodes in $M$. We show that this happens with constant probability for sufficiently large $m/\ell$.

We first bound, in terms of $m/\ell$, the probability that the $i^{th}$ distinct node in $M$ will copy the first node's signal the first time it's chosen.

\begin{claim}\label{claim.ithnode}
Let $v_i$ denote the $i^{th}$ distinct node of $M$ chosen in $G_m^\ell$. Let also $\ell_i$ denote the number of distinct nodes in $L_{v_i}$ that have been chosen when $v_i$ is first chosen. Then, for all $i \geq 2$,
$Pr[\ell_i \geq i-1] \leq  e^{-(i-1)\frac{m+\ell}{4m}}$
and
$Pr[\ell_i \geq i-1] \leq 2\frac{\ell}{m+\ell}$.
\end{claim}

\begin{proof}
Sample the nodes chosen by the process in the following way: first, sample the order in which distinct nodes in $M$ will be revealed (i.e. pick an ordering of $M$ uniformly at random). Then, each time a new node is chosen, first sample whether or not it's a new node in $M$. If so, choose the next element of $M$ to be revealed. If not, pick a node uniformly at random from $L$ and the nodes in $M$ that have already been revealed. Then it is clear that $\ell_i \geq i-1$ if and only if $i-1$ distinct nodes of $L_{v_i}$ are chosen before $i$ distinct nodes of $M$ are chosen. So consider the $j^{th}$ distinct node in $M \cup L_{v_i}$ revealed, and let $X_j$ be the random variable that is $1$ if the node is from $M$, and $0$ if it's from $L_{v_i}$. Then clearly, either $i$ distinct nodes in $M$ or $i-1$ distinct nodes in $L_{v_i}$ have been revealed by the time $2i-2$ distinct nodes in $M \cup L_{v_i}$ have been revealed. So we see that $\ell_i \geq i-1$ if and only if
$\sum_{j=1}^{2i-2} X_j \leq i-1.$
We show now that this has low probability by using a Chernoff bound. It is easy to see that the set of $X_j$s are negatively correlated, and also that $\mathbb{E}[X_j] = \frac{m}{m+\ell}$. So $\mathbb{E}[\sum_{j=1}^{2i-2} X_j] = \frac{m(2i-2)}{m+\ell}$. So, by a Chernoff bound, 
\begin{align*}
Pr\left[\sum_{j=1}^{2i-2} X_j \leq i-1\right] &\leq e^{-(\frac{m+\ell}{2m})^2\frac{2m(i-1)}{m+\ell}/2}
\leq e^{-(i-1)\frac{m+\ell}{4m}}.
\end{align*}
The second part of the claim is easier to prove: it's clear that $Pr[\ell_i \geq i-1] \leq Pr[\ell_2 \geq 0]$ for all $i$. It is also easy to see that $Pr[\ell_2 \geq 0]$ is at most $2\frac{\ell}{m+\ell}$ (because $\ell_2 \geq 0$ if and only if $X_1$ or $X_2$ is $1$).
\end{proof}

\begin{corollary}\label{cor.allnodes}
For any choice of $j > 1$, with probability at least $1 - \frac{2(j-2)\ell}{m+\ell} - \frac{e^{-(j-1)\frac{m+\ell}{4m}}}{1-e^{-\frac{m+\ell}{4m}}} - \frac{\ell}{m+\ell}$, whenever any node in $M$ is chosen for the first time in $G_m^\ell$, it copies the signal of the first node chosen from $M$.
\end{corollary}

\begin{proof}
$\sum_{i = j}^\infty e^{-(i-1)\frac{m+\ell}{4m}}$ is a geometric sum with ratio $e^{-\frac{m+\ell}{4m}}$, so $\sum_{i = j}^\infty e^{-(i-1)\frac{m+\ell}{4m}} = \frac{e^{-(j-1)\frac{m+\ell}{4m}}}{1-e^{-\frac{m+\ell}{4m}}}$. By Claim~\ref{claim.ithnode} and a union bound, the probability that $\ell_i > i-1$ for any $v_i$, $i \geq j$ is at most $ \frac{e^{-(j-1)\frac{m+\ell}{4m}}}{1-e^{-\frac{m+\ell}{4m}}}$. Also by Claim~\ref{claim.ithnode} and a union bound, the probability that $\ell_i > i-1$ for any $2 \leq i < j$ is at most $2(j-2)\frac{\ell}{m+\ell}$. Lastly, the probability that the first node in $M$ to be chosen is chosen before any of its leaves is exactly $\frac{\ell}{m+\ell}$. In the event that this happens, it will clearly report its own signal. Taking a union bound over all three events proves the corollary.
\end{proof}

At this point we have shown that when nodes in $M$ are chosen for the first time, they are likely to copy the opinion of the first node chosen. We show now that whenever this happens, no node in $M$ will change their report when they are chosen again.

\begin{observation}\label{obs.nochange}
Consider any sequence of the process on $G_m^\ell$ such that every node in $M$ copies the opinion of the first node chosen in $M$ the first time they are chosen. Then every node in $M$ will continue to copy this opinion if they are chosen again later.
\end{observation}

\begin{proof}
Once a node $v \in M$ makes an announcement, every leaf in $L_v$ that announces after $v$ will copy $v$. As $v$ copied the first node of $M$ the first time it was chosen, along with every other node in $M$, this means that no more neighbors of $v$ will disagree with this opinion when $v$ is chosen again, but more nodes may agree.
\end{proof}

Plugging in $j = 50$ and $m/\ell = 200$ gives a bound of at least $1/3$ in Corollary~\ref{cor.allnodes}. It is also easy to see that the average degree of $G_{200\ell}^\ell$ is no more than $201$. Therefore, as we let $\ell$ grow to $\infty$, we can get an arbitrarily large graph of constant average degree that arrives at a blue consensus with non-negligible probability. In fact, the graph will reach a blue majority after only a linear number of iterations of the process, which ``violates'' Proposition~\ref{prop:make this a proposition} (i.e. there is no way to update the constants in Proposition~\ref{prop:make this a proposition} and replace max degree with average degree to obtain a true statement).

\begin{claim}
The probability of a blue consensus in $G_{200\ell}^\ell$ is at least $1/6 - \delta/3$ for all $\ell > 0$. Furthermore, as $\ell \rightarrow \infty$, the probability that there is a blue majority after $2000\ell^2$ rounds approaches $1/6 - \delta/3$.
\end{claim}

\begin{proof}
The first part of the claim is an immediate corollary of Corollary~\ref{cor.allnodes} and Observation~\ref{obs.nochange} (and plugging in $j = 50$, $m/\ell = 200$). The second part of the claim comes from the observation that after $1000 \ell^2$ rounds, it is extremely likely (with probability approaching $1$ as $n \rightarrow \infty$) that at least $3/4$ of the nodes in $M$ have been chosen. Therefore, with probability approaching $1/6 - \delta/3$, after $1000 \ell^2$ rounds, there will be at least $150\ell$ nodes in $M$ who have announced blue, and none who have announced red. From this point, any node adjacent to these $150\ell$ nodes that is chosen will also announce blue. As this is over $3/4$ of the entire graph, it is also extremely likely that more than $2/3$ of these nodes will be chosen in the next $1000 \ell^2$ rounds. As each such chosen node will announce blue if chosen, we would reach a blue majority.
\end{proof}